\documentclass[12pt]{article}

\usepackage[pdftex]{graphicx}

\usepackage{amsmath}
\usepackage{amssymb}
\usepackage{amsfonts}
\usepackage{amsthm}

\usepackage{latexsym}

\usepackage{xspace}
\newtheorem{theorem}{Theorem}[section]
\newtheorem{proposition}[theorem]{Proposition}

\newenvironment{myeps}[3]%
        {\begin{figure}[#1]%
        \begin{center}%
        \resizebox{#2}{!}{%
                \includegraphics{#3}}%
        \end{center}}%
        {\end{figure}}
        
\newcommand{\cpc}{competitive parallel computing\xspace}
\newcommand{\Cpc}{Competitive parallel computing\xspace}

\begin{document}

\title{A sufficient condition for a linear speedup \\in competitive parallel computing}

\author{Naoki Yonezawa\thanks{Information System Course, Department of Business Administration, Faculty of Modern Life, Teikyo Heisei University, Nakano, Tokyo 164--8530, Japan. Email: \texttt{n.yonezawa@thu.ac.jp}}}

\date{}

\maketitle

\begin{abstract}
In competitive parallel computing, the identical copies of a code in a phase of a sequential program are assigned to processor cores and the result of the fastest core is adopted. In the literature, it is reported that a superlinear speedup can be achieved if there is an enough fluctuation among the execution times consumed by the cores. \Cpc is a promising approach to use a huge amount of cores effectively. However, there is few theoretical studies on speedups which can be achieved by \cpc at present. In this paper, we present a behavioral model of \cpc and provide a means to predict a speedup which \cpc yields through  theoretical analyses and simulations. We also found a sufficient condition to provide a linear speedup which \cpc yields. More specifically, it is sufficient for the execution times which consumed by the cores to follow an exponential distribution. In addition, we found that the different distributions which have the identical coefficient of variation (CV) do not always provide the identical speedup. While CV is a convenient measure to predict a speedup, it is not enough to provide an exact prediction.
\end{abstract}


\section{Introduction}
Multi-core and many-core are in the mainstream of parallel computing and there is a steady increase in the number of their cores. However, in the near future, it is expected that the degree of parallelism is below the number of cores which the hardware provides due to the restriction of the problems to solve or the algorithms to execute~\cite{hill2008amdahl}. Meanwhile, it is getting more and more difficult to write a parallel program because 1) it is necessary to control a huge amount of flows of program execution and 2) the elements of parallel computing system get diversified over the last decade. In the future, writing a parallel program gets complicated extremely as the number of cores grows~\cite{pankratius2008software,patterson2013computer}.

To alleviate these above problems, \textit{competitive parallel computing} or its equivalent are proposed~\cite{Cledat:2011:ESU:2048066.2048109,Ertel1992,Trachsel:2010:VCP:1787275.1787325}. In competitive parallel computing, the identical copies of a code in a phase of a sequential program are assigned to processor cores and the result of the fastest core is adopted. It is reported that a superlinear speedup can be achieved 
if there is an enough fluctuation among the execution times consumed by the cores. \Cpc has some advantages over conventional \textit{cooperative parallel computing}; 1) in \cpc, it is not necessary to parallelize an existing program, and 2) \cpc is applicable to algorithms which are impossible or difficult to parallelize. \Cpc is a promising approach to use a huge amount of cores effectively.

However, there is few theoretical studies on speedups which can be achieved by \cpc at present. Specifically, although it is intuitively understood that a larger fluctuation among the execution times, namely, a larger coefficient of variation (CV) provides a larger speedup, to the best of our knowledge, the relation between a CV and a speedup is not evaluated quantitatively in detail.

The contributions of this paper are following:

\begin{itemize}
\item We present the behavioral model of \cpc and provide a means to predict a speedup which \cpc yields through  theoretical analyses and simulations.
\item We found a sufficient condition to provide a linear speedup which \cpc yields. More specifically, it is sufficient for the execution times which consumed by the cores to follow exponential distribution. In addition, we proved that exponential distribution is not the only distribution to achieve a linear speedup. In other words, exponential distribution is not a necessary condition to achieve a linear speedup.
\item We found that the different distributions which have the identical CV do not always provide the identical speedup. While CV is a convenient measure to predict a speedup, it is not enough to provide an exact prediction.
\end{itemize}

The rest of this paper is organized as follows: Section \ref{s:related_work} provides related work. In Section \ref{s:mathematical_analysis}, we propose a mathematical model to evaluate \cpc and present a means to calculate the execution time of \cpc. In Section \ref{s:evaluation}, we evaluate speedups which \cpc provides through Monte Carlo simulation based on our proposed model and present the relation between a CV and a speedup quantitatively. Finally, we conclude our study and describe our future work in Section \ref{s:conclusion}.

\section{Related Work}
\label{s:related_work}

The ideas of assigning the undivided computation and making processors compete among them are proposed in the literature.

Wolfgang~\cite{Ertel1992} proposes random competition, in which the computations compete using the randomness in search algorithm. Although he analyzes speedups based on the variance of the measured execution times, there is no mention of CV.

Trachsel \textit{et al.}~\cite{Trachsel:2010:VCP:1787275.1787325,trachsel2008platform} propose Computation-driven CPE (Competitive Parallel Execution) and Compiler-driven CPE. To achieve a speedup, Computation-driven CPE assigns the different algorithm or the identical algorithm with the different initial conditions to processors while Compiler-driven CPE yields the combinations of the optimization options for compiler to generate variant codes. In experimental evaluation, they obtained superlinear speedups with specific data sets. However, they do not link the superlinear speedups with CV.

Without enough attention to the degree of the variance of the execution time among processors,  using \cpc naively wastes computing resources. To overcome this problem, Cledat \textit{et al.}~\cite{Cledat:2011:ESU:2048066.2048109,Cledat:2009:OCN:1855591.1855596,cledat2010energy} proposes the methods called \textit{learning} and \textit{culling}. The CV of WalkSAT, one of the application they adopted for evaluation, is less than one and the speedup  is worse than a linear speedup. Meanwhile, the CV of another application, namely, MSL motion planning is greater than one and a superlinear speedup is achieved. These results are consistent with our result. Therefore, it is proper to claim that our results reinforce and extend their work.

\section{A Mathematical Analysis of Competitive Parallel Computing}
\label{s:mathematical_analysis}

In this section, we show a behavioral model of \cpc and describe how to calculate the execution time based on the model.

\subsection{A model of the program execution in competitive parallel computing}

In general, a sequential program consists of several \textit{phases}. In this paper, we define the term `a phase' as a program region which is between two semantic points. For example, a phase is a sentence of a program, a function call, or an iteration of a loop. Typically, the result which is produced in a phase is consumed in the consecutive or later phases.

While sequential computing executes a phase on a single core, \cpc assigns the identical copies of the phase to multiple cores and makes the cores compete. The result of the fastest core is adopted and the other cores are terminated. Then, the flow of control goes to the next phase. In this paper, to model \cpc which behaves as mentioned above, we consider the minimum model which represents a program consisting of a single phase and running on $n$ cores as shown in Figure \ref{fig:cpc}.

\begin{myeps}{bt}{75mm}{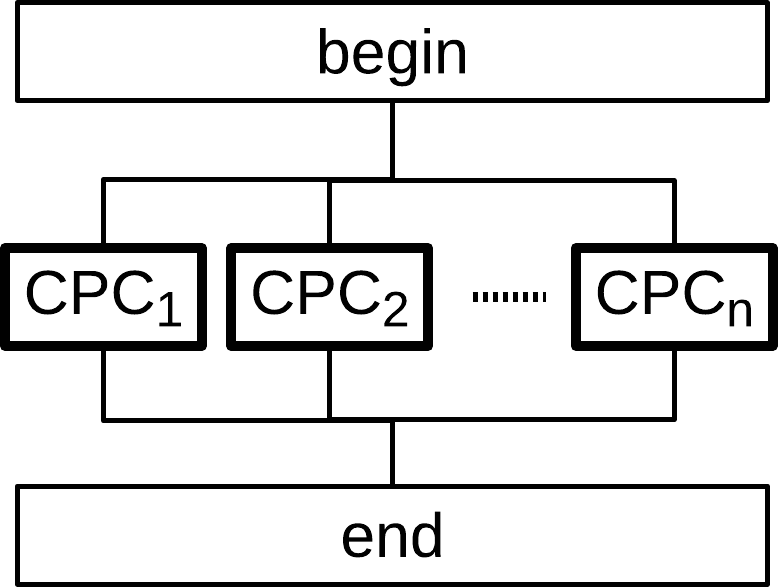}
\caption{A minimum model of competitive parallel computing (CPC)}
\label{fig:cpc}
\end{myeps}

The execution times of a phase running on the different cores might be different each other if the cores are assigned to the different algorithms or the identical algorithm with the different parameters. The external factors including cache misses and network delay also produce the fluctuation of the execution times. These cause \textit{randomness}. In order to model the execution of such a program, we denote the execution time of Core $i$ by $X_i$, $i = 1, 2, \dots, n$, where $\{X_i\}_{i=1}^n$ are independent and identically distributed random variables (i.i.d. r.v.'s). At this time, the overall execution time of the program is $Y_n=\min(X_1, X_2, \dots, X_n)$ because the execution time of the program is the execution time of the fastest core. In the next section, we obtain the probability distribution which the random variable $Y_n$ follows.

\subsection{Calculations of the execution time}

As mentioned above, we assume that $n$ random variables $X_1, X_2, \dots, X_n$ is i.i.d. r.v.'s. We denote the cumulative distribution function (CDF) which these $n$ random variables follow by

\[
F_X(x)=P(X_1\le x)=P(X_2\le x)=\cdots=P(X_n\le x)
\]
and also denote its probability density function (PDF) (if exists) by $f_X(x)$.

\begin{proposition}
The random variable $Y_n$ is the minimum among $n$ random variables $X_1, X_2, \dots, X_n$, that is, $Y_n=\min(X_1, X_2, \dots, X_n)$. The CDF which $Y_n$ follows is as follows:
\[
F_{Y_n}(y)=1-(1-F_X(y))^n.
\]
\end{proposition}
\begin{proof}
\begin{eqnarray}
F_{Y_n}(y)&=&P(Y_n\le y)\nonumber\\
&=&P(\min(X_1, X_2, \dots, X_n)\le y)\nonumber\\
&=&1-P(\min(X_1, X_2, \dots, X_n)> y)\nonumber\\
&=&1-P(X_1>y, X_2>y, \dots, X_n> y)\nonumber\\
&=&1-P(X_1>y)P(X_2>y)\cdots P(X_n>y)\nonumber\\
&&(\because \text{independent random variables})\nonumber\\
&=&1-(1-P(X_1\le y))(1-P(X_2\le y))\cdots(1-P(X_n\le y))\nonumber\\
&=&1-(1-F_X(y))(1-F_X(y))\cdots(1-F_X(y))\nonumber\\
&=&1-(1-F_X(y))^n
\label{eq:min_cdf}
\end{eqnarray}
\end{proof}

\section{Evaluation}
\label{s:evaluation}

We calculated speedups which \cpc provides through Monte Carlo simulation based on our proposed model. We gave four probability distributions to the model.

In this section, we describe the four probability distribution~\cite{fushimi2004e,takahashi2008e}. Then, we demonstrate the exact solution in case that random variables which represent the execution time follows exponential distribution. Finally, we show the results of the simulation and discuss the relation between a CV and a speedup.  

\subsection{Probability distribution of the execution time}
\label{s-sec:distribution}
The overall execution time may vary significantly depending on the distribution of random variables which represent the execution time for each processor core. In this study, we assume that a random variable follows one of four distributions: exponential distribution, Erlang distribution, hyperexponential distribution, and uniform distribution. 

In exponential distribution whose parameter is $\lambda$, the CDF for $X_i$ is assumed in the form
\begin{equation}\label{eq:exponential-cdf}
F_X(x) = P(X_i \le x) = 1 - e^{-\lambda x}
\end{equation}
for $i = 1, 2, \dots, n$, so that the PDF is in the form 
\begin{equation}\label{eq:exponential-pdf}
f_X(x) = \lambda e^{-\lambda x}
\end{equation}
and expected value (mean execution time) becomes $E(X_i)=\frac{1}{\lambda}$.

The CDF and the PDF of Erlang distribution are
\begin{equation*}\label{eq:erlang-cdf}
F_X(x)=1-e^{-\lambda k x}\sum_{r=0}^{k-1}\frac{(\lambda k x)^r}{r!},
\end{equation*}
\begin{equation*}\label{eq:erlang-pdf}
f_X(x)=\frac{(\lambda k)^k}{(k-1)!}x^{k-1}e^{-\lambda k x},
\end{equation*}
respectively, where $k$ is the number of phases\footnotemark.
\footnotetext{The term of $phase$ which is used in the context of probability theory is unrelated to a phase which is included in a program.}
The expected value of random variables which follow the above Erlang distribution is also $\frac{1}{\lambda}$.

The CDF and the PDF of hyperexponential distribution are
\[
F_X(x)=1-\sum_{j=1}^kC_je^{-\lambda_jx},
\]
\begin{equation*}\label{eq:hyper-pdf}
f_X(x)=\sum_{j=1}^kC_j\lambda_je^{-\lambda_jx},
\end{equation*}
respectively, where $\{C_j\}_{j=1}^k$ is an arbitrary discrete distribution. 
We chose parameters of hyperexponential distribution so that all of their expected value is equal to $\frac{1}{\lambda}$ as in the above two distributions. As a result, we obtained the PDF of hyperexponential distribution as follows:
\begin{equation}\label{eq:hyper-pdf-with-concrete-parameters}
f_X(x)=\frac{a}{2}\lambda e^{-a\lambda x}+\frac{a}{4a-2}\lambda e^{-\frac{a}{2a-1}\lambda x},
\end{equation}
where $a (\neq \frac{1}{2})$ is a real number.

The adoption of these distributions for the execution time is based on the following idea. For non-negative random variables with the same expected value, 
CV is the most useful and popular characteristic parameter for comparing the degree of variation. The CV $c(X)$ for non-negative random variable $X$ is defined by
\[
c(X) = \frac{\sqrt{V(X)}}{E(X)}
\]
where $V(X)$ is variance of $X$, i.e., $V(X) = E(X^2)-E(X)^2$. It is clear that for fixed value of $E(X)$, as increases the value of $c(X)$, the variance of $X$ also increases. In the field of probability theory, exponential distribution, Erlang distribution, and hyperexponential distribution are the most typical distributions with different CV. It is well known that $c(X) = 1$ for exponential distribution, $c(X) < 1$ for Erlang distribution, and $c(X) > 1$ for hyperexponential distribution. In other words, for the same expected value, Erlang distribution shows lower variance and hyperexponential distribution shows higher variance comparing with exponential distribution. 

In this paper, we additionally adopt uniform distribution. The CDF and the PDF of uniform distribution are
\[
F_X(x)=P(X_i\le x)=\frac{x-a}{b-a},
\]
\[
f_X(x)=\frac{1}{b-a},
\]
respectively, where $0 \le a < X_i \le b$.
The CV of uniform distribution is less than one, that is, 
\[
c(X_i)=\frac{b-a}{\sqrt{3}(b+a)}<1.
\]

\subsection{The exact solution with exponential distribution}

We show the exact solution of the expected execution time in case that the $n$ random variables follow exponential distribution.
Hereafter, we assume that $\lambda=1$ without loss of generality.

\subsubsection{the execution time for $n=1$}

In general, the expected value of a random variable $X$ is calculated as follows:

\begin{equation}\label{eq:expected-value}
E(X)=\int_0^\infty xf(x)dx,
\end{equation}
where $f(x)$ is PDF which the random variable follows.

For $n=1$, $Y_1=\min(X_1)=X_1$. Therefore, using Equation \ref{eq:exponential-pdf} and Equation \ref{eq:expected-value},

\[
E(Y_1)=\int_0^\infty x \lambda e^{-\lambda x}dx=\frac{1}{\lambda}=1.
\]

\subsubsection{the execution time for $n>1$}
\label{sss:n>1}
Hereafter, we define a speedup as $S_n=E(Y_1)/E(Y_n)$.

\begin{theorem}
(a) `random variables which represent the execution times of processor cores follow  exponential distribution independently' is a sufficient condition for (b) 'achieving a linear speedup', but is not a necessary condition.
\end{theorem}
\begin{proof}
\mbox{}\begin{description}
\item[(a) is a sufficient condition for (b):] 
Using Equation \ref{eq:min_cdf} and Equation \ref{eq:exponential-cdf}, 
\[
F_{Y_n}(y)=1-e^{-\lambda ny}, 
\] 
\begin{equation}\label{eq:fyn-exponential}
f_{Y_n}(y)=\lambda n e^{-\lambda ny}.
\end{equation}
Therefore,
\[
E(Y_n)=\int_0^\infty yf_Y(y)dy=\frac{1}{\lambda n}=\frac{1}{n}.
\]
Consequently, $S_n=n$, that is, a linear speedup.

\item [(a) is not a necessary condition for (b):]
It is sufficient to show another distribution which provides (b).
If the random variable $Y_n$ follows the distribution which is represented as Equation \ref{eq:hyper-pdf-with-concrete-parameters} with $n\lambda$ instead of $\lambda$, namely,
\[
f_{Y_n}(x)=\frac{a}{2}n\lambda e^{-an\lambda x}+\frac{a}{4a-2}n\lambda e^{-\frac{a}{2a-1}n\lambda x},
\]
then $E(Y_n)=\frac{1}{\lambda n}=\frac{1}{n}$. This is another example which provides (b) and is different to Equation \ref{eq:fyn-exponential}
which is obtained from exponential distribution. Therefore, it is proved that (a) is not a necessary condition for (b).

\end{description}
\end{proof}

\subsection{The results of numerical experiments}
To evaluate the relation between a CV and a speedup, we calculated the execution times with varying the distributions which the random variables follow. We carried out Monte Carlo simulations as shown in Algorithm 1.

\newpage
\noindent
\underline{\textbf{Algorithm 1:}}
\begin{description}
\item[Input: ] the number of steps $N$, the distribution $D$, the number of processor cores $n$.
\item[Output: ] the execution time.
\end{description}
\begin{enumerate}
\item $i \leftarrow 0$. 
\item $S \leftarrow 0$.
\item Substitute the random numbers which follows the distribution $D$ into the $n$ random variables $X_1, X_2, \dots, X_n$~\cite{morimatsui2004e}.
\item $Y_n\leftarrow\min(X_1, X_2, \dots, X_n)$.
\item $S \leftarrow S + Y_n$.
\item $i \leftarrow i + 1.$
\item if $i < N$ go to Step 3, otherwise go to Step 8.
\item Output $\frac{S}{N}$.
\end{enumerate}
\vspace{2mm}

We varied the number of cores $n = \{1, 2, \dots, 100\}$ and defined $N$ as 100,000. 
We denote exponential distribution, Erlang distribution with parameter $k$, and hyperexponential distribution with parameter $a$ by M, $\mathrm{E}_{k}$, $\mathrm{H}_{2(a)}$, respectively, derived from Kendall's notation in queuing theory~\cite{gross2008fundamentals}.

\subsubsection{Comparing Speedups among various CVs}
\label{sss:varying_CV}
The speedups with varying CV are shown in Figure \ref{fig:speedup-100-proc-various-cvs}. CVs are shown in Table \ref{tab:cv}.
\begin{table}[bt]
\caption{Coefficients of variation (CV)}
\hbox to\hsize{\hfil
\begin{tabular}{c|c|c|c|c}\hline\hline
Distribution $D$ & $\mathrm{E}_{10}$ & M & $\mathrm{H}_{2(5)}$ & $\mathrm{H}_{2(10)}$\\\hline
Coefficients of variation & 0.32 & 1.00 & 1.51 & 1.62\\
\end{tabular}\hfil}
\label{tab:cv}
\end{table}

\begin{myeps}{bt}{140mm}{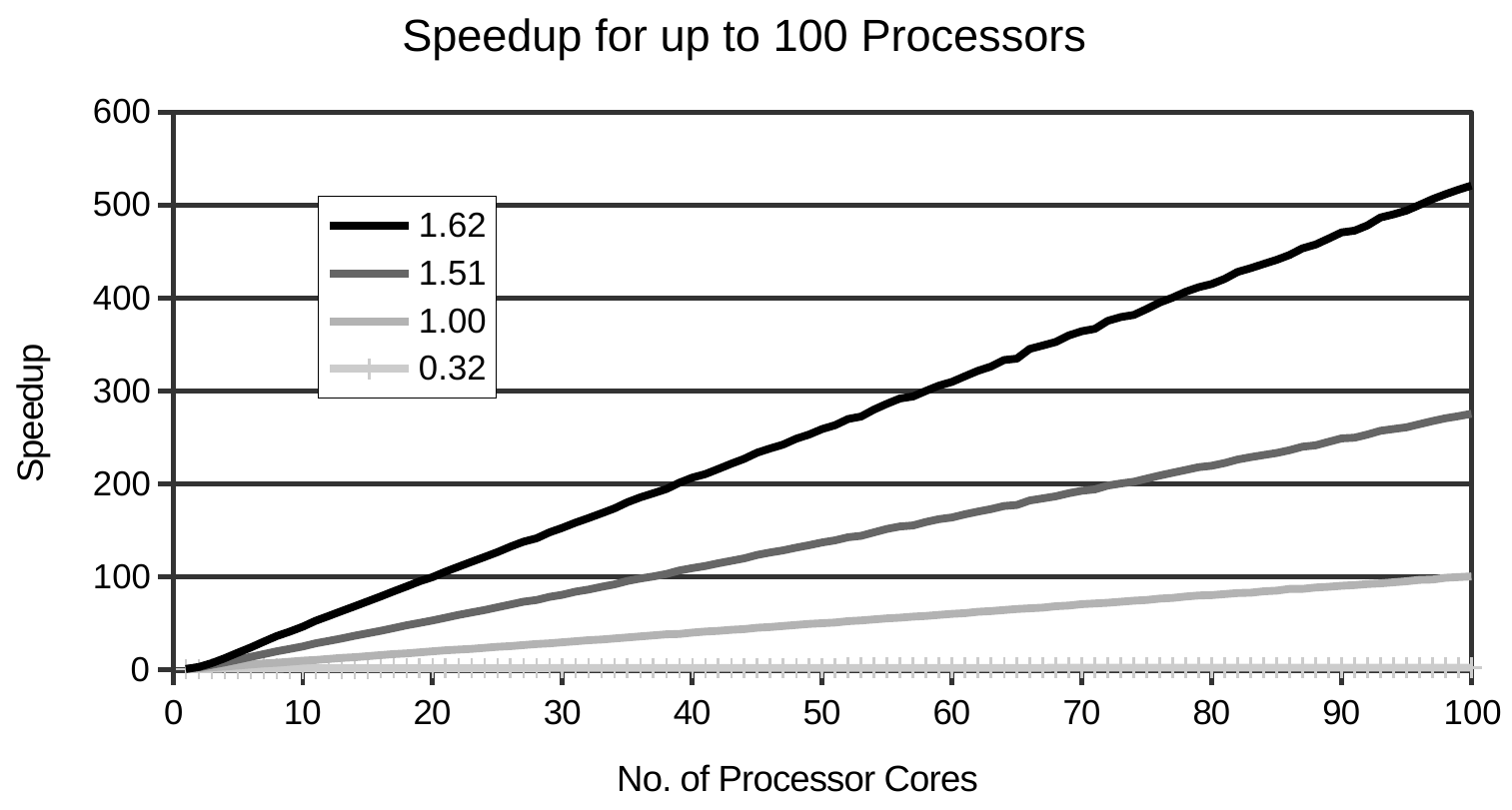}
\caption{Speedup for various CVs}
\label{fig:speedup-100-proc-various-cvs}
\end{myeps}

As a whole, more cores and a larger CV bring a larger speedup. From these results, we theoretically confirmed the fact which is intuitively predicted and is confirmed experimentally by other studies. With the distribution M, a linear speedup is achieved, which is identical to the exact solution mentioned in Section \ref{sss:n>1}. With the distribution $\mathrm{H}_{2(5)}$ (CV is 1.51), a speedup is 275.71 for $n = 100$. With the distribution $\mathrm{H}_{2(10)}$ (CV is 1.62), a speedup is 521.15 for $n = 100$. These results show sufficient conditions for achieving superlinear speedups while these hyperexponential distributions might not reflect the behavior of a real-world application.

\begin{myeps}{bt}{140mm}{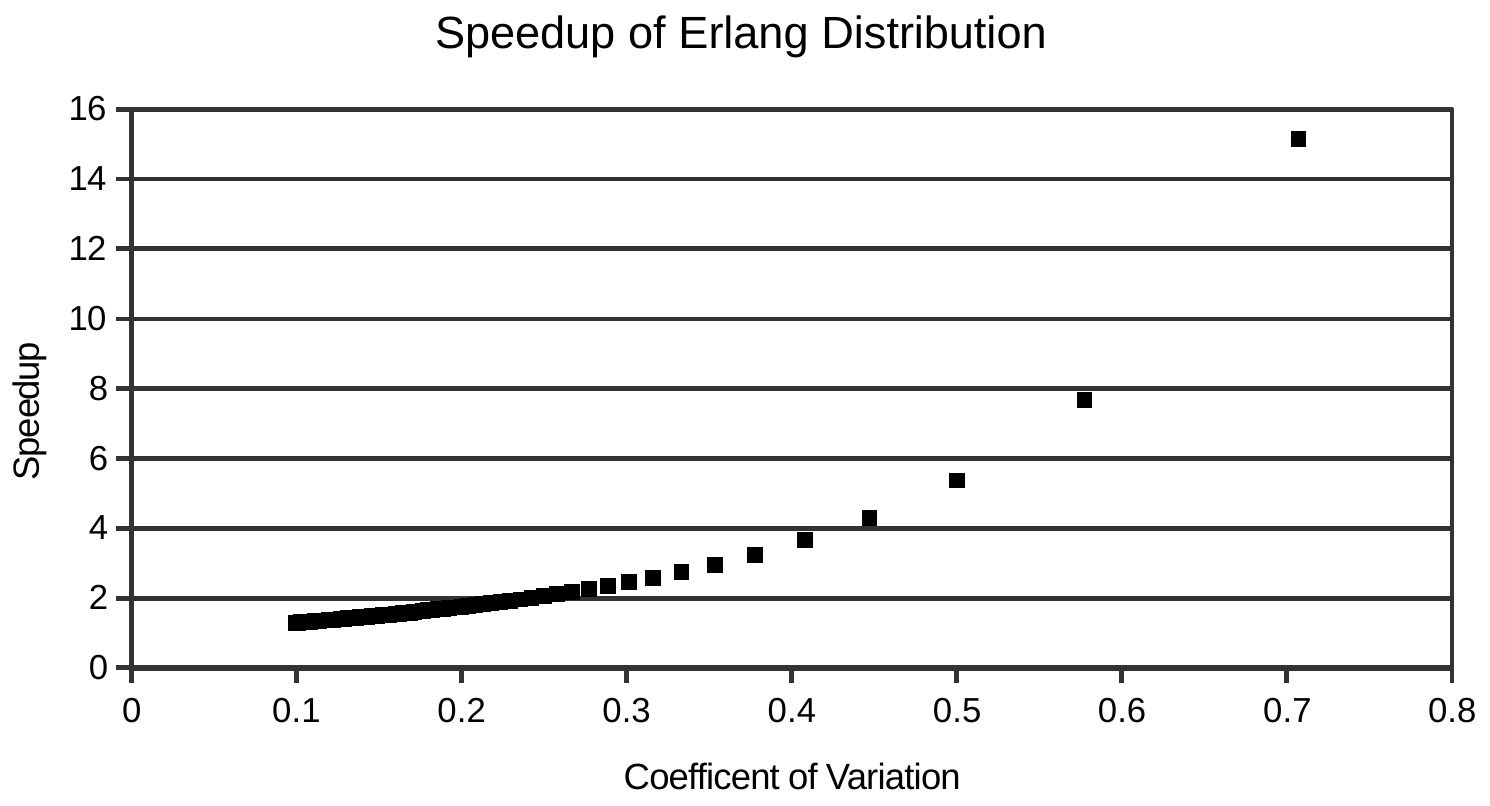}
\caption{Speedup for Erlang Distribution ($n=100$)}
\label{fig:speedup-erlang}
\end{myeps}

\subsubsection{Speedups for one hundred processors with extreme CVs}
In Section \ref{sss:varying_CV}, we found that the different CV provides the different speedup. To explore the relation between a CV and a speedup in more detail, we carried out simulations with varying CV more finely for $n=100$. 

We show the speedups with varying the parameter $k$ of Erlang distribution 2 to 100 in Figure \ref{fig:speedup-erlang}. Speedups are 7.68 and 15.14 for 0.58 and 0.71 as CV, respectively. These speedups are possibly acceptable as performance gains for $n=100$. Meanwhile, the speedups are lower than 2 with lower CVs. These speedups are unacceptable unless computing resources and electricity are abundantly available.

We show the speedups with varying the parameter $a$ of hyperexponential distribution 1 to 100 in Figure \ref{fig:speedup-hyper}. Note that hyperexponential distribution is equivalent to exponential distribution if $a=1$. While the speedup is 426.08 for 1.59 as CV, the speedup grows rapidly, that is, as 1,798.56 and 4,975.12 for 1.70 and 1.72 as CV, respectively. If someone finds an application which shows such behavior, a huge performance gain is obtained. Although it might not be realistic to find such an application, it is meaningful to obtain a theoretical perspective for a performance gain.

\begin{myeps}{bt}{140mm}{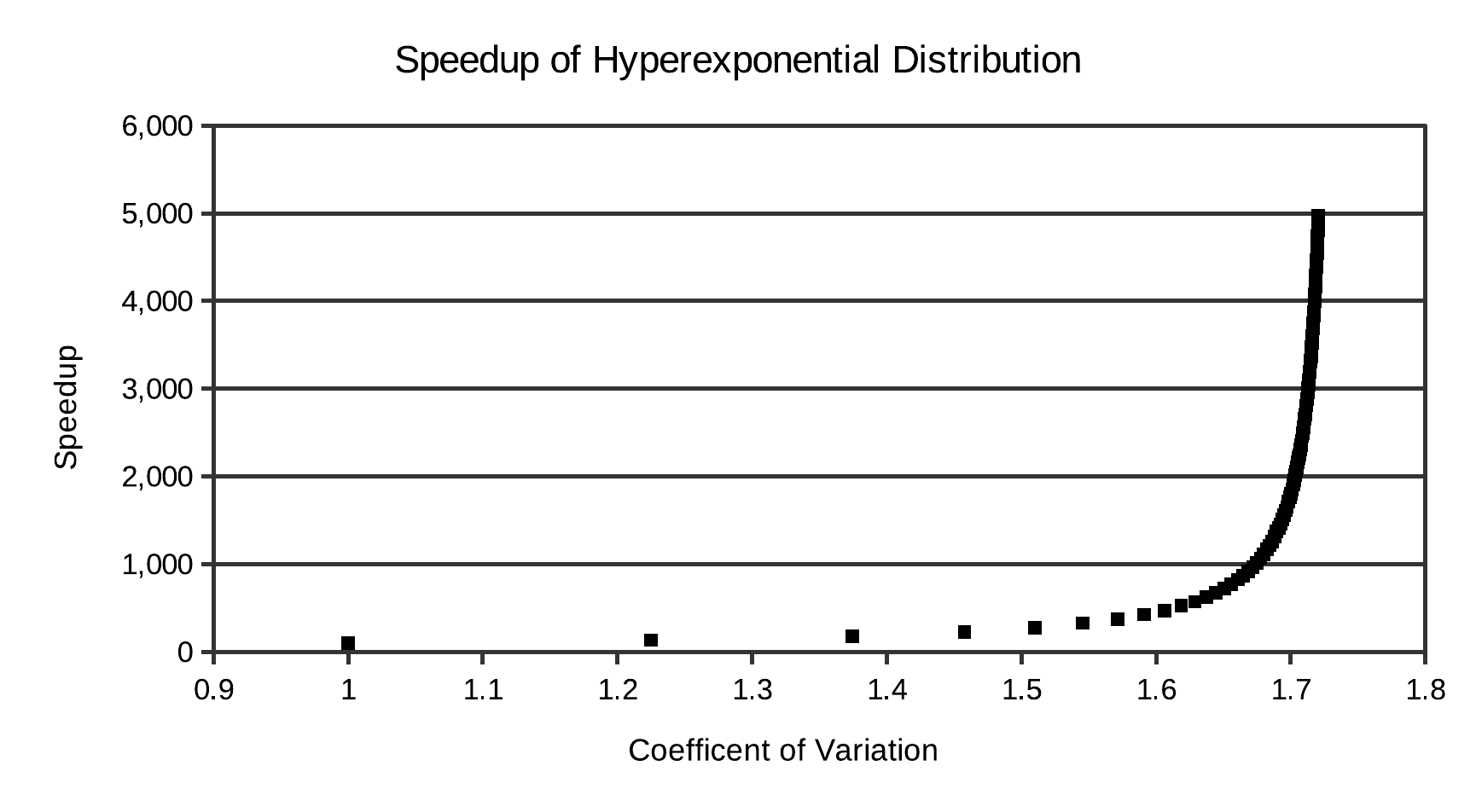}
\caption{Speedup for Hyperexponetial Distribution ($n=100$)}
\label{fig:speedup-hyper}
\end{myeps}

\subsubsection{Comparing Speedups when fixing CV}

Finally, we compare the speedups for Erlang distribution with the ones for uniform distribution, to explore the speedups with the identical CV provided by different distributions. We adopted the parameter $k$ as $3$ for Erlang distribution and chose the parameters as $a=0, b=2$ for uniform distribution so that the expected value is 1, which is the same as other distributions. As a result, CV is $\frac{1}{\sqrt{3}}\approx 0.58$ for both Erlang distribution and uniform distribution.

\begin{myeps}{bt}{140mm}{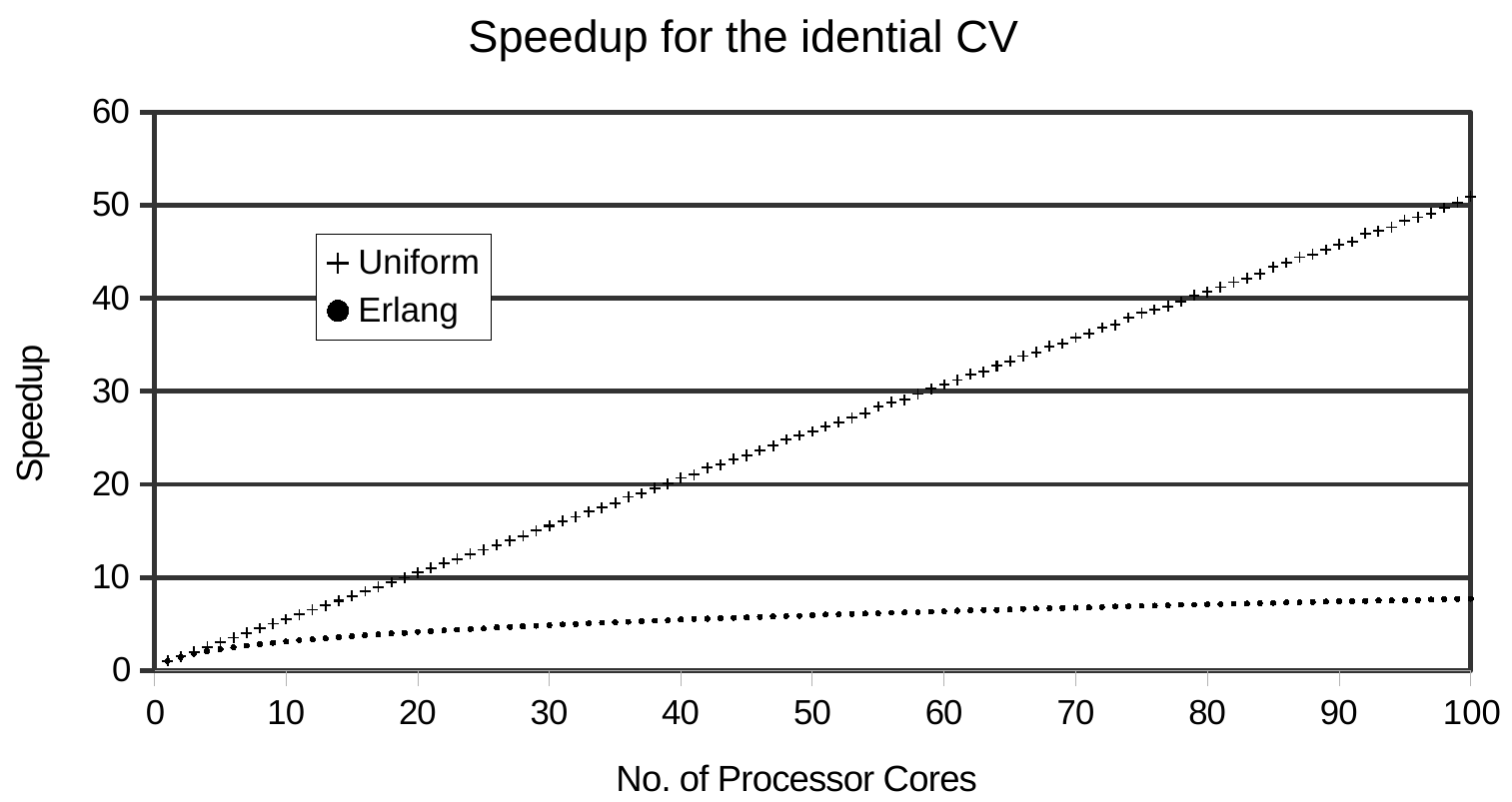}
\caption{Comparing speedups for the identical CV}
\label{fig:speedup-100-proc-idential-cv}
\end{myeps}

We show the speedups with varying the number of cores 1 to 100 in Figure \ref{fig:speedup-100-proc-idential-cv}. The speedups for uniform distribution are better than the ones for Erlang distribution. The reason for these results is as followings: the shape of graph corresponding to PDF of Erlang distribution is mountain type while the shape corresponding to PDF of uniform distribution is horizontally flat. In other words, the random number which follows Erlang distribution tends to be around the peak of the graph while the probability of generating a random number for uniform distribution is equal to any other numbers. Therefore, the probability for generating a smaller number for uniform distribution is higher than the one for Erlang distribution. In our model, because the smallest random variable which shows the shortest execution time among cores is adopted as the overall execution time, the speedups for uniform distribution are better than the ones for Erlang distribution.

These results show that the identical CV does not always yield the identical speedup. While CV is a convenient measure to predict a speedup, it is necessary to consider the distribution as well as CV for a more precise prediction.

\section{Conclusion}
\label{s:conclusion}
In this paper, we constructed a mathematical model which represents the behavior of \cpc and theoretically analyzed \cpc using the model. We investigate sufficient conditions which provide a linear speedup through a theoretical analysis as well as simulations with various kinds of probability distribution. As a consequence, we found that exponential distribution yields a linear speedup and is not the only distribution which yields a linear speedup. This imply that it is possible to find the different distribution which yields a linear distribution and is easier to be realized as a real-world entity than exponential distribution.

Although CV is consider as a convenient measure to predict a speedup so far, we found that the identical CV does not always yield the identical speedup through the experiments with the fixed CV.

Our future work will include:
\begin{itemize}
\item to find a \textit{better} distribution which yields a linear or superlinear speedup than exponential distribution. In other words, such a distribution should be easier to realized as a real-world entity than exponential distribution.
\item to evaluate our proposed model using applications. More specifically, we compare the predicted speedup which is obtained through a probabilistic analysis of an application with the corresponding measured speedup experimentally.
\end{itemize}

\end{document}